\documentclass{article}

\usepackage{mathptmx}
\usepackage[scaled=.9]{helvet}
\usepackage{type1cm} % Make sure to use non-fixed size fonts
\usepackage{courier}

\usepackage[all]{xy}
\usepackage{graphicx}
\usepackage[font=small]{caption}
\usepackage{subcaption}
\usepackage{amsmath}
\usepackage{amsthm}
\usepackage{amsfonts}
\usepackage{amsthm}
\usepackage{url}
\usepackage{float}
\usepackage{wrapfig}

\allowdisplaybreaks

\usepackage{xcolor}
\usepackage{proof}
\usepackage{listings}
\usepackage{fancyvrb}
\usepackage{relsize}

\theoremstyle{plain}
\newtheorem{theorem}{Theorem}
\newtheorem{lemma}[theorem]{Lemma}
\newtheorem{corollary}[theorem]{Corollary}
\theoremstyle{definition}
\newtheorem{notation}[theorem]{Notation}
\newtheorem{convention}[theorem]{Convention}
\newtheorem{definition}[theorem]{Definition}
\newtheorem{remark}[theorem]{Remark}
\newtheorem{example}[theorem]{Example}

\newcommand{\PCF}{\mbox{${\bf PCF}^{\bf list}$}}

\newcommand{\bit}{{\tt bit}}
\newcommand{\bool}{{\tt bool}}

\newcommand{\cnotgate}[2]{{\tt N}({#1}\cdot{#2})}
\newcommand{\notgate}[1]{{\tt N}({#1})}
\newcommand{\wires}[1]{{\it Wires}({#1})}

\newcommand{\define}[1]{{\em #1}}

\newcommand{\bor}{\,{|}\,}

\newcommand{\listtype}[1]{{[}{#1}{]}}
\newcommand{\unittype}{{\tt 1}}
\newcommand{\tunit}{\unittype}

\newcommand{\lettermx}[3]{{{\it let}~{#1}~=~{#2}~{\it in}~{#3}}}
\newcommand{\punit}{\mathtt{skip}}%{{\star}}
\newcommand{\prodterm}[1]{{\langle{#1}\rangle}}
\newcommand{\letprodterm}[3]{{{\it let}~\prodterm{#1}~=~{#2}~{\it in}~{#3}}}
\newcommand{\letunitterm}[2]{{#1}\mathtt{;}{#2}}%{{\mathtt{let}~\punit~=~{#1}~\mathtt{in}~{#2}}}

\newcommand{\ttrue}{\mathtt{t\!t}}
\newcommand{\ffalse}{\mathtt{f\!f}}

\newcommand{\iftermx}[3]{{{\mathtt{if}}~{#1}~\mathtt{then}~{#2}~\mathtt{else}~{#3}}}

\newcommand{\inj}{{\tt inj}}
\newcommand{\injl}{\inj_1}
\newcommand{\injr}{\inj_2}
\newcommand{\match}[5]{{{\mathtt{match}}~{#1}~{\mathtt{with}}~({#2\,\mapsto\,#3}\vert{#4\,\mapsto\,#5})}}
\newcommand{\letrec}[4]{{{\mathtt{letrec}}~{#1}\,{#2}={#3}~{\mathtt{in}}~{#4}}}

\newcommand{\errorlist}{{\tt Err}}

\newcommand{\nil}{{\mathtt{nil}}}
\newcommand{\cons}[3][]{{{#2}\,{\mathtt :}{\mathtt :}^{#1}\,{#3}}}

\newcommand{\splitlist}[1][]{{\mathtt{split}^{#1}}}

\newcommand{\emptyam}{\textrm{{\it EmptyAM}}}
\newcommand{\rwam}{\mathrel{\to_{\it am}}}
\newcommand{\cbvam}{\mathrel{\to_{\it cbv}}}
\newcommand{\monadreturn}{{\tt return}}
\newcommand{\monadapp}{{\tt app}}

\newcommand{\Lift}[1]{{\it Lift}_{\mathcal{#1}}}

\newcommand{\wiretype}{{\tt wire}}
\newcommand{\monadttrue}{{\tt mtt}}
\newcommand{\monadffalse}{{\tt mff}}
\newcommand{\monadif}{{\tt mif}}
\newcommand{\monadand}{{\tt mand}}
\newcommand{\monadxor}{{\tt mxor}}
\newcommand{\monadnot}{{\tt mnot}}

\newcommand{\gatetype}{{\tt gate}}

\newcommand{\statetype}{{\tt state}}
\newcommand{\circtype}{{\tt circ}}

\newcommand{\succterm}{{\bf S}}

\newcommand{\void}[1]{}

\title{Generating reversible circuits from higher-order functional
  programs}

\author{Beno\^{i}t Valiron}

\date{March 20, 2016}

\begin{document}

\maketitle

\begin{abstract}
Boolean reversible circuits are boolean circuits made of reversible
elementary gates. Despite their constrained form, they can simulate
any boolean function. The synthesis and validation of a reversible
circuit simulating a given function is a difficult problem. In 1973,
Bennett proposed to generate reversible circuits from traces of
execution of Turing machines.
In this paper, we propose a novel presentation of this approach,
adapted to higher-order programs. Starting with a PCF-like language,
we use a monadic representation of the trace of execution to turn a
regular boolean program into a circuit-generating code. We show that a
circuit traced out of a program computes the same boolean function as
the original program.
This technique has been successfully applied to generate large oracles
with the quantum programming language Quipper.
\end{abstract}

\section{Introduction}
\label{sec:introduction}

Reversible circuits are linear, boolean circuits with no loops, whose
elementary gates are reversible. 
In quantum computation, reversible circuits are mostly used as oracle:
the {\em description} of the problem to solve. Usually, this
description is given as a classical, conventional algorithm: the graph
to explore~\cite{TF}, the matrix coefficients to process~\cite{QLSA},
{\em etc}. These algorithms use arbitrarily complex structures, and if some
are rather simple, for example~\cite{shor}, others are quite complicated and make use of
analytic functions~\cite{QLSA}, memory registers~\cite{BF} (which
thus have to be simulated), {\em etc}.

This paper\footnote{A shorter preprint has been  accepted for publication in the Proceedings of {\em Reversible Computation 2016}.
The final publication is available at
  \url{http://link.springer.com}.} is concerned with the design of reversible circuits as {\em
  operational semantics} of a higher-order purely functional
programming language. The language is
 expressive enough to encode most algorithms: it features
recursion, pairs, booleans and lists, and it can easily be extended
with additional structures if needed. This operational semantics can be
understood as the {\em compilation} of a program into a reversible
circuit.

Compiling a program into a reversible circuit is fundamentally
different from compiling on a regular back-end: there is no notion of
``loop'', no real control flow, and all branches will be explored
during the execution.  In essence, a reversible circuit is the {\em
  trace} of all possible executions of a given program. Constructing a
reversible circuit out of the trace of execution of a program is what
Bennett~\cite{bennett73logical} proposed in 1973, using Turing
machines. In this paper, we refer to it as Landauer
embeddings~\cite{landauer61irreversibility}.

In this paper, we build up on this idea of circuit-as-trace-of-program
and formalize it into an operational semantics for our higher-order
language. This semantics is given externally as an abstract machine,
and internally, as a {\em monadic interpretation}.

The strength of our approach to circuit synthesis is to be able to
reason on a regular program independently from the constraints of the
circuit construction. The approach we follow is similar to what is
done in Geometry of synthesis~\cite{ghica07gos} for hardware synthesis, but
since the back-end we aim at is way simpler, we can devise a very
natural and compact monadic operational semantics.

\smallskip
\noindent
{\bf Contribution.}
The main contribution of this paper is a monadic presentation of
Landauer embeddings~\cite{landauer61irreversibility} in the context of
higher-order programs. Its main strength is its parametricity: a
program really represents a {\em family} of circuits, parametrized on
the size of the input. Furthermore, we demonstrate a compositional
monadic procedure for generating a reversible circuit out of a
regular, purely functional program. The generated circuit is then
provably computing the same thing as the original program. This can be
used to internalize the generation of a reversible circuit out of a
functional program. It has been implemented in Quipper~\cite{Quipper}
and used for building complex quantum oracles.

\smallskip
\noindent
{\bf Related works.}
From the description of a conventional function it is always possible
to design a reversible circuit computing the function out of its
truth table or properties thereof and several methods have been
designed to generate compact
circuits (see e.g.~\cite{rev-survey-2011,maslov-templates-iccad03,synth4,synth13,synth2,synth17}).
However, if these techniques allow one to write
reversible functions with arbitrary truth tables~\cite{revlib}, they
do not usually scale well as the size of the input grows.

Synthesis of reversible circuits can be seen as a small branch of the
vast area of hardware synthesis. In general, hardware synthesis can be
structural (description of the structure of the circuit) or behavioral
(description of algorithm to encode). Functional programming languages
have been used for both. On the more structural side one finds
Lava~\cite{Claessen-2001}, BlueSpec~\cite{bluespec}, functional
netlists~\cite{park08functional}, {\em etc}.  On the behavioral side
we have the Geometry of Synthesis~\cite{ghica07gos},
Esterel~\cite{esterel}, ForSyDe~\cite{forsyde}, {\em etc}. Two more
recent contributions sitting in between structural and behavioral
approaches are worth mentioning. First, the imperative, {\em
  reversible} synthesis language SyRec~\cite{syrec}, specialized for
reversible circuits. Then, Thomsen's proposal~\cite{thomsen}, allowing
to represent a circuit in a functional manner, highlighting
the behavior of the circuit out of its structure.

On the logic side, the geometry of interaction~\cite{goi} is a
methodology that can be used to turn functional programs into
reversible computation~\cite{abramsky,ghica07gos,terui}: it is based on
the idea of turning a typing derivation into a reversible automaton.

There have also been attempts to design reversible abstract machines
and to compile regular programs into reversible computation. For
example, a reversible version of the SEMCD machine has been
designed~\cite{kluge99}. More recently, the compiler REVS~\cite{revs}
aims at compiling conventional computation into reversible circuits.

Monadic semantics for representing circuits is something relatively
common, specially among the DSL community: Lava~\cite{Claessen-2001},
Quipper~\cite{PLDI}, Fe-Si~\cite{braibant13formal}, {\em etc}. Other
approaches use more sophisticated constructions, with type systems
based on arrows~\cite{James} in order to capture reversibility. 

In the present work, the language is circuit-agnostic, and the
interest of the method lies more in the fact that the monadic
semantics to build reversible circuits is completely {\em implicit}
and only added at circuit-generation time, following the approach
in~\cite{MLmonad}, rather than in the choice of the language. Compared
to~\cite{James}, our approach is also parametric in the sense that a
program does not describe one fixed-size circuit but a family of
circuits, parametrized by the size of the input.

\smallskip
\noindent
{\bf Plan of the paper.}
Section~\ref{sec:reversible-circuits} presents the definition of
reversible circuits and how to perform computation with them.
Section~\ref{sec:progr-revers-circ} describes a PCF-like
lambda-calculus and proposes two operational semantics: one as a
simple beta-reduction and one using an abstract machine and a partial
evaluation procedure generating a
circuit. Section~\ref{sec:internalizing} describes the call-by-value
reduction strategy and explains how to internalize the
abstract-machine within the language using a
monad. Section~\ref{sec:use-cases} discusses the use of this technique
in the context of the generation of quantum oracles, and discusses the
question of optimizing the resulting circuits.  Finally,
Section~\ref{sec:conclusion} concludes and proposes some future
investigations.

\section{Reversible circuits}
\label{sec:reversible-circuits}

\begin{wrapfigure}{r}{0.3\textwidth}
\vspace{-10pt}
\includegraphics[width=0.28\textwidth]{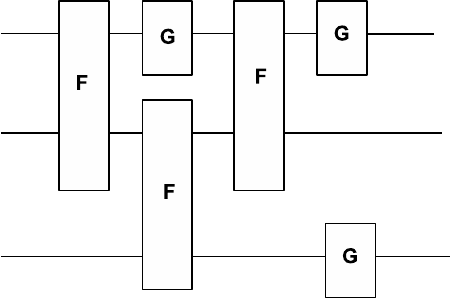}
\vspace{-10pt}
\end{wrapfigure}
A reversible boolean circuit consists in a set of \define{open wires}
and \define{elementary gates} attached onto the wires. Schematically,
a reversible boolean circuit is of the form shown on the right.
To each gate is associated a boolean operation, supposed to be
reversible. In this circuit example, $G$ is a one-bit operation (for
example a not-gate, flipping a bit) while $F$ is a two-bit
operation. In each wire, a bit ``flows'' from left to right. All the
bits go at the same pace. When a gate is met, the corresponding
operation is applied on the wires attached to the gate. Since the
gates are reversible, the overall circuit is reversible by making the
bits flow backward. 

\paragraph{\bf Choice of elementary gates.}

Many gates have been considered in the
literature~\cite{rev-survey-2011}. In this paper, we will consider
multi-controlled-not gates. A not gate, represented by
$
\xymatrix@R=.05in@C=.1in{
  \ar@{-}[r] & *{\oplus} \ar@{-}[r] &~,
}
$
is flipping the value of the wire on which it is attached. The
operator ${\tt not}$ stands for the bit-flip operation. Given a
gate $F$ acting on $n$ wires, a controlled-F is a gate acting on $n+1$
wires. The control can be positive or negative, represented
\begin{wrapfigure}{r}{0.35\textwidth}
\vspace{-10pt}
\begin{minipage}{0.33\textwidth}
  $\xymatrix@R=.05in@C=.1in{
  x\ar@{-}[r] & *{\bullet} \ar@{-}[r]\ar@{-}[d] & x
  \\
  \vec{y}\ar@{=}[r] & *+[F]{F} \ar@{=}[r] & 
}
\quad
\xymatrix@R=.05in@C=.1in{
  x \ar@{-}[r] & *{\circ} \ar@{-}[r]\ar@{-}[d] & x
  \\
  \vec{y} \ar@{=}[r] & *+[F]{F} \ar@{=}[r] &
}$
\end{minipage}
\vspace{-10pt}
\end{wrapfigure}
respectively as shown on the right.
In both cases, the top wire is not modified. On the bottom wires, the
gate $F$ is applied if $x$ is true for the positive control, and
false for the negative control. Otherwise, no gate is applied: the
values $\vec{y}$ flow unchanged through the gate.
A positively-controlled not gate will be denoted CNOT.

A reversible circuit runs a computation on some query: some input
wires correspond to the query, and some output wires correspond to the
answer. The auxiliary input wires that are not part of the query are
initially fed with the boolean ``false'' (also written $0$).

\paragraph{\bf Computing with reversible circuits.}
As described by Landauer\,\cite{landauer61irreversibility} and
Bennett\,\cite{bennett73logical}, a conventional, classical algorithm
that computes a boolean function $f:\bit^n \to \bit^m$ can be mechanically
transformed into a reversible circuit sending the triplet
$(x,\vec{0},\vec{0})$ to $(x,{\rm trace},f(x))$, as in
Figure~\ref{fig:Tf}. Its input wires are not modified by the circuit,
and the trace of all intermediate results are kept in garbage wires.

Because of their particular structure, two Landauer embeddings $T_g$ and
$T_h$ can be composed to give a Landauer embedding of
the composition $h\circ g$. Figure~\ref{fig:TgTh} shows the process:
the wires of the output of $T_g$ are fed to the input of $T_h$, and
the output of the global circuit is the one of $T_h$. The garbage
wires now contain all the ones of $T_g$ and $T_h$.

Note that it is easy to build elementary Landauer embeddings for
negation and conjunction: the former is a negatively-controlled not
while the latter is a positively doubly-controlled not. 
Any boolean function can then be computed with Landauer embeddings.

\def\multiline#1#2#3#4{\save
(#1,#4#2);(#3,#2)**\dir{-};
(#1,#4#2.9999);(#3,#4#2.9999)**\dir{-};
(#1,#4#2.5);(#3,#4#2.5)**\dir[|(0.5)]{-};
(#1,#4#2.25);(#3,#4#2.25)**\dir{-};
(#1,#4#2.75);(#3,#4#2.75)**\dir{-};
(#1,#4#2.625);(#3,#4#2.625)**\dir[|(0.5)]{-};
(#1,#4#2.375);(#3,#4#2.375)**\dir[|(0.5)]{-};
(#1,#4#2.55);(#3,#4#2.55)**\dir[|(0.5)]{-};
(#1,#4#2.45);(#3,#4#2.45)**\dir[|(0.5)]{-};
%(#1,#4#2.46875);(#3,#4#2.46875)**\dir[|(0.1)]{-};
%(#1,#4#2.53125);(#3,#4#2.53125)**\dir[|(0.1)]{-};
\restore}
\begin{figure}[tb]
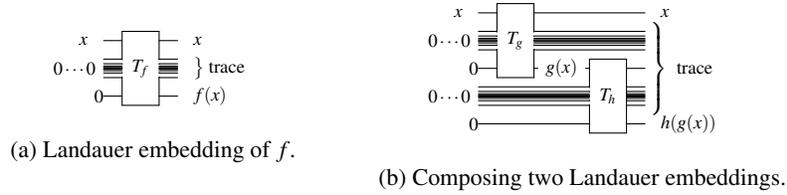

\scriptsize\centering
\begin{subfigure}{.35\textwidth}
  \centering
  $
  \begin{array}{c}
    \xy
    0;<1em,0em>:
    (0,0.5);(2,0.5)**\dir{-};
    (2,2)**\dir{-};
    (2,3)**\dir{.};
    (2,4.5)**\dir{-};
    (0,4.5)**\dir{-};
    (0,3)**\dir{-};
    (0,2)**\dir{.};
    (0,0.5)**\dir{-};
    (1,2.5)*{T_f};
    (0,1);(-1,1)*!R{0}**\dir{-};
    (0,4);(-1,4)*++!R{x}**\dir{-};
    (2,1);(3,1)*++!L{f(x)}**\dir{-};
    (2,4);(3,4)*++!L{x}**\dir{-};
    (-1,2.5)*+!R{{0\cdots0}};
    (3,2.5)*++!L{\begin{array}{@{}c@{}}\big\}~\textrm{trace}\end{array}}
    \multiline{0}{2}{-1}{}
    \multiline{2}{2}{3}{}
    \endxy
  \end{array}
  $
  \caption{Landauer embedding of $f$.}
  \label{fig:Tf}
\end{subfigure}
\qquad
\begin{subfigure}{.5\textwidth}
\centering$
\begin{array}{c}
\xy
0;<1em,0em>:
(0,0.5);(2,0.5)**\dir{-};
(2,2)**\dir{-};
(2,3)**\dir{.};
(2,4.5)**\dir{-};
(0,4.5)**\dir{-};
(0,3)**\dir{-};
(0,2)**\dir{.};
(0,0.5)**\dir{-};
(1,2.5)*{T_g};
(0,1);(-1,1)*!R{{0}}**\dir{-};
(0,4);(-1,4)*++!R{x}**\dir{-};
(2,1);(3.5,1)*+{g(x)}**\dir{-};(5,1)**\dir{-};
(2,4);(8,4)*++!L{x}**\dir{-};
(-1,2.5)*+!R{{0\cdots0}};
(5,-2.5);(7,-2.5)**\dir{-};
(7,-1)**\dir{-};
(7,-1)**\dir{.};
(7,1.5)**\dir{-};
(5,1.5)**\dir{-};
(5,0)**\dir{-};
(5,-1)**\dir{.};
(5,-2.5)**\dir{-};
(6,-0.5)*{T_h};
(5,-2);(-1,-2)*!R{{0}}**\dir{-};
(7,-2);(8,-2)*++!L{h(g(x))}**\dir{-};
(7,1);(8,1)*++++!L{\textrm{trace}}**\dir{-};
(-1,-0.5)*+!R{{0\cdots0}};
\multiline{0}{2}{-1}{}
\multiline{2}{2}{8}{}
\multiline{5}{-0}{-1}{}
\multiline{7}{-0}{8}{}
\save
(8.3,3).(8.3,-1)!C*+\frm{\}}
\restore
\endxy
\end{array}
$
\caption{Composing two Landauer embeddings.}
\label{fig:TgTh}
\end{subfigure}
\caption{Landauer embeddings.}
\label{fig:class-shape-to-closed}
\end{figure}

\section{Reversible circuits as trace of programs}
\label{sec:progr-revers-circ}

In this section, we present an implementation of Landauer embeddings
to the context of a higher-order functional programming language, and
show how it can be understood through an abstract machine.

\subsection{Simple formalization of reversible circuits}

A reversible circuit has a very simple structure. As a linear sequence
of elementary gates, it can be represented as a simple list of gates.

\begin{definition}\rm
  A \define{reversible gate} $G$ is a term 
  $
    \cnotgate{i}{b_1^{j_1}\ldots b_n^{j_n}}
  $
  where $i$, $j_1$,\ldots,$j_n$ are natural numbers such that for all
  $k$, $i\neq j_k$, and where $b_1$,\ldots,$b_n$ are booleans. If the list of
  $b_k^{j_k}$ is empty, we simply write $\notgate{i}$ in place of
  $\cnotgate{i}{}$.
  The \define{wires} of the gate
  $\cnotgate{i}{b_1^{j_1}\ldots b_n^{j_n}}$ is the set of natural
  numbers $\{i,j_1,\ldots, j_n\}$. The wire $i$ is called
  \define{active} and the $j_k$'s are called the \define{control
    wires}.
  Given a list $C$ of gates, the union of
  the sets of wires of the elements of $C$ is written $\wires{C}$.
  Finally, the boolean values True and False flowing in the wires are
  respectively represented with $\ttrue$ and $\ffalse$ throughout the paper.
\end{definition}

\begin{definition}
  A \define{reversible boolean circuit} is a triplet $(I,C,O)$ where
  $C$ is a list of reversible gates and where $I$ and $O$ are sets of
  wires. The list $C$ is the \define{raw circuit}, $I$ is the set of
  \define{inputs wires} and $O$ the set of \define{outputs wires}. We
  also call $\wires{C}\setminus I$ the \define{auxiliary wires} and
  $\wires{C}\setminus O$ the \define{garbage wires}.
\end{definition}

Executing a reversible circuit on a given tuple of booleans
computes as follows.

\begin{definition}\label{def:exec-circ}
  Consider a circuit $(I,C,O)$ and a family of bits $(x_i)_{i\in I}$.
  A \define{valuation} for the circuit is an indexed family
  $v=(v_j)_{j\in\wires{C}\cup I\cup O}$ of booleans.
  The \define{execution of a gate $\cnotgate{i}{b_1^{j_1}\ldots b_n^{j_n}}$
    on the valuation $v$} is the valuation $w$ such that for all
  $l\neq i$, $w_l=v_l$ and
  $
  w_i = v_i~{\tt xor}~\wedge_{k=1}^n(v_{j_k}~{\tt xor}~b_k~{\tt xor}~\ttrue)
  $
  if $n\geq 1$ and $w_i = {\tt not}(v_i)$ otherwise.
  The execution of the circuit $(I,C,O)$ with input $(x_i)_{i\in I}$ is
  the succession of the following operations:
  (1) Initialization of a valuation $v$ such that for all $k\in I$,
    $v_k=x_k$, and for all the other values of $k$, $v_k$ is false.
  (2) Execution of every gate in $C$ on $v$, {\em in reverse order}.
  (3) The execution of the circuit returns the sub-family
    $(v_k)_{k\in O}$.
\end{definition}

\subsection{A PCF-like language with lists of booleans}

In this section, we present the functional language \PCF{} that we use
to describe the regular computations that we eventually want to
perform with a reversible circuit. The language is simply-typed and it
features booleans, pairs and lists.
\[
\begin{array}{l@{~~}l@{~~}l}
  M,N & {:}{:}{=} &
  x \bor
  \lambda x.M \bor
  MN \bor
  \prodterm{M,N} \bor
  \pi_1(M) \bor
  \pi_2(M) \bor
  \punit \bor
  \letunitterm{M}{N} \bor
  \ttrue \bor
  \ffalse \bor
  \\
  &&
  \iftermx{M}{N}{P} \bor
  {\tt and} \bor {\tt xor} \bor {\tt not} \bor
  \inj_1(M) \bor \inj_2(M) \bor
  \\
  &&
  \match{P}{x}{M}{y}{N} \bor
  \splitlist^A \bor
  Y(M)
  \bor\errorlist,
  \\[1ex]
  A,B & {:}{:}{=} &
  \bit \bor A \oplus B \bor A\times B \bor 
  \unittype \bor A \to B \bor \listtype{A}.
\end{array}
\]
The language comes equipped with the typing rules of
Table~\ref{tab:typ-rules}. There are several things to note.  First,
the construct $\iftermx{\!\textrm{-}\!}{\!\textrm{-}\!}{}$
can only output \define{first-order
  types}. A first order type is a type from the
grammar
$
A^0,B^0~ {:}{:}{=} ~ \bit \bor A^0\times B^0 \bor \listtype{A^0}.
$
Despite the fact that one can encode them with the test construct, for
convenience we
add the basic boolean combinators ${\tt not}$, ${\tt
  xor}$ and ${\tt and}$.
There are no constructors for lists, but instead
there is a coercion from $\tunit\oplus(A\times\listtype{A})$ to
$\listtype{A}$; the term $\splitlist$ turns a list-type into a
additive type. There is a special-purpose term $\errorlist$ that
will be used in particular in Section~\ref{sec:op-sem-rev-circ} as an
error-spawning construct. The boolean values True and False are
respectively represented with $\ttrue$ and $\ffalse$. $\punit$ is the unit term and
$\letunitterm{M}{N}$ is used as the destructor of the unit. 
Finally, $Y$ is a fixpoint operator. As we shall eventually work with
a call-by-value reduction strategy, we only consider fixpoints
defining functions.

\begin{table}[t]
  {\[
  \infer{\Delta,x:A\vdash x:A}{}
  \quad
  \infer{\Delta\vdash \ttrue:\bit}{}
  \quad
  \infer{\Delta\vdash \ffalse:\bit}{}
  \quad
  \infer{\Delta\vdash \punit:\tunit}{}
  \quad
  \infer{\Delta\vdash \errorlist:A}{}
  \]
  \[
  \infer{\Delta\vdash {\tt not}:\bit\to\bit}{}
  \quad
  \infer{\Delta\vdash {\tt and}:\bit\times\bit\to\bit}{}
  \quad
  \infer{\Delta\vdash {\tt xor}:\bit\times\bit\to\bit}{}
  \]
  \[
  \infer{\Delta\vdash \splitlist:\listtype{A}\to\tunit\oplus(A\times\listtype{A})}{}
  \]
  \[
  \infer{\Delta\vdash\lambda x.M : A\to B}{
    \Delta,x:A\vdash M:B
  }
  \qquad
  \infer{\Delta\vdash\pi_i(M):A_i}{
    \Delta\vdash M:A_1\times A_2
  }
  \qquad
  \infer{\Delta\vdash\inj_i(M):A_1\oplus A_2}{
    \Delta\vdash M:A_i
  }
  \]
  \[
  \infer{\Delta\vdash MN : B}{
    \Delta\vdash M : A \to B
    &
    \Delta\vdash N : A
  }
  \qquad
  \infer{\Delta\vdash \prodterm{M,N} : A\times B}{
    \Delta\vdash M : A
    &
    \Delta\vdash N : B
  }
  \qquad
  \infer{\Delta\vdash \letunitterm{M}{N} : B}{
    \Delta\vdash M : \tunit
    &
    \Delta\vdash N : B
  }
  \]
  \[
  \infer{\Delta\vdash \match{P}{x^A}{M}{y^B}{N} : C}{
    \Delta\vdash P : A\oplus B
    &
    \Delta,x:A\vdash M : C
    &
    \Delta,y:B\vdash N : C
  }
  \qquad
  \infer{\Delta\vdash M:\listtype{A}}{
    \Delta\vdash M:\tunit\oplus(A\times\listtype{A})
  }
  \]
  \[
  \infer{\Delta\vdash Y(M):A}{
    \Delta\vdash M : A\to A
  }
  \qquad
  \infer{\Delta\vdash \iftermx{P}{M}{N} : C}{
    \Delta\vdash P : \bit
    &
    \Delta\vdash M : C
    &
    \Delta\vdash N : C
    &
    \textrm{the type $C$ is first-order}
  }
  \]}
  \caption{Typing rules of \PCF{}.}
  \label{tab:typ-rules}
\end{table}

\begin{notation}\label{notation:syntax}\rm
  We write $\nil$ for $\injl(\punit)$ and $\cons MN$ in place of $\injr (M\times N)$. We
  also write $[M_1,\ldots M_n]$ for
  $\cons{M_1}{\cons{\ldots}{\cons{M_n}{\nil}}}$.  We also write
  general products $\prodterm{M_1,\ldots,M_n}$ as
  $\prodterm{M_1,\prodterm{\ldots M_n\ldots}}$. Projections $\pi_i$
  for $i\leq n$ extends naturally to $n$-ary products. We write
  $\letrec{f}{x}{M}{N}$ for the term $(\lambda f.N)(Y(\lambda f.\lambda x.M))$.
\end{notation}

\begin{remark}
  \label{rem:if1storder}
  The typing rule of the {\tt if}-{\tt then}-{\tt else} construct
  imposes a first-order condition on the branches of the test. This
  will be clarified in Remark~\ref{rem:if1storderbis}. For now, let us
  just note that this constraint can be lifted with some syntactic
  sugar: if $M$ and $N$ are of type $A_1\to\ldots\to A_n$, where $A_n$
  is first-order, then a ``higher-order'' test $\iftermx{P}{M}{N}$ can
  be defined using the native first-order test by an $\eta$-expansion
  with  the lambda-abstraction
  $\lambda x_1\ldots
  x_n.\iftermx{P}{Mx_1\ldots x_n}{Nx_1\ldots x_n}$.
\end{remark}

\subsection{Small-step semantics}
\label{sec:betared}

\begin{table}[tb]
  \[
    \begin{array}{c}
    \begin{array}{r@{}lr@{}lr@{}l}
    (\lambda x.M)N &\to M{[}N/x{]}
    &
    \pi_i\prodterm{M_1,M_2} &\to M_i
    &
    \letunitterm{\punit}{M}
    &\to M
    \\[1ex]
    \iftermx{\ttrue}{M}{N}
    &\to M
    &
    \iftermx{\ffalse}{M}{N}
    &\to N
    &
    \splitlist\,M &\to M
    \end{array}
    \\[2.5ex]
    \begin{array}{r@{}lr@{}l}
    \match{\inj_i(P)}{x_1}{M_1}{x_2}{M_2}
    &\to
    M_i{[}P/x_i{]}
    &
    Y(M)
    &\to
    M(Y(M))
    \end{array}
    \end{array}
    \]
  \caption{Small-step semantics for \PCF{}: reduction rules, acting on
    subterms.}\label{tab:betared}
\end{table}

We equip the language \PCF{} with the smallest rewrite-system closed
under subterm reduction, satisfying the rewrite rules of
Table~\ref{tab:betared}, and satisfying the obvious rules regarding
${\tt not}$, ${\tt and}$ and ${\tt xor}$: for example, ${\tt not}~\ttrue\to
\ffalse$ and ${\tt not}~\ffalse\to \ttrue$.  Note that the term
$\errorlist$ does not reduce. This is on purpose: it represents an
error that one cannot catch with the type system; in particular it
will be used in Section~\ref{sec:op-sem-rev-circ}.
The usual safety properties are satisfied, modulo the error-spawning
term $\errorlist$.

\begin{definition}
  A value is a term $V$ defined by the grammar $\lambda x.M\bor
  \prodterm{U_1,U_2}\bor \inj_i(U) \bor c$, where $c$ is
  a constant term: $\punit$, $\ttrue$, $\ffalse$.
\end{definition}

\begin{theorem}[Safety]
  Type preservation and progress are verified:
  (1) If $\Delta\vdash M:A$, then for all $N$ such that $M\to N$ we
    also have
    $\Delta\vdash N:A$.~~
  (2) If $M$ is a closed term of type $A$ then either $M$ is a
    value, or $M$ contains the term $\errorlist$, or $M$
    reduces.\qed
\end{theorem}

In summary, the language is well-behaved. It is also reasonably
expressive, in the sense that most of the computations that one could
want to perform on lists of bits can be described, as shown in
Example~\ref{ex:list}.

\begin{convention}\rm
  When defining a large piece of code, we will be using a Haskell-like
  notation. So instead of defining a closed function as a lambda-term
  on a typing judgment, we shall be using the notation
\begin{Verbatim}[fontsize=\relsize{-2}]
function : type_of_the_function
function arg1 arg2 ... = body_of_the_function
\end{Verbatim}
Also, we shall use the convenient notation $\lettermx{x}{M}{N}$ for
$(\lambda x.N)M$ and the notation $\letprodterm{x,y}{\!M}{N}$ for
$\lettermx{z\!\!}{\!\!M}{\lettermx{x\!}{\!\pi_1(z)}{\lettermx{y\!}{\!\pi_1(z)}{N}}}$.
Similarly, we allow multiple variables for recursive functions, and we
use pattern-matching for lists and general products in the same
manner.
\end{convention}

\begin{example}[List combinators]\label{ex:list}
  The usual list combinators can be defined. 
Here we give the definition of {\tt foldl: (A $\to$ B $\to$ A) $\to$ A
  $\to$ [B] $\to$ A}. The other ones (such as map, zip\ldots) are written
similarly.
\begin{Verbatim}[fontsize=\relsize{-2},commandchars=\\\{\},
  codes={\catcode`$=3\catcode`^=7}]
foldl f a l = letrec g z l' = match (split l') with 
                     nil$\hspace{2.1ex}$ $\mapsto$ z
                   | $\langle$h,t$\rangle$ $\mapsto$ g (f z h) t
              in g a l
\end{Verbatim}
\end{example}

\begin{example}[Ripple-carry adder]\label{ex:adder}
  One can easily encode a bit-adder: it takes a carry and two
  bits to add, and it replies with the answer and the carry to forward.
\begin{Verbatim}[fontsize=\relsize{-2},commandchars=\\\{\},
  codes={\catcode`$=3\catcode`^=7}]
bit_adder : bit $\to$ bit $\to$ bit $\to$ (bit $\times$ bit)
bit_adder carry x y =
      let majority a b c = if (xor a b) then c else a in
      let z = xor (xor carry x) y in
      let carry' = majority carry x y in $\langle$carry', z$\rangle$
\end{Verbatim}
  Encoding integers as lists of bits, low-bit first, one can use the
  bit-adder to write a complete adder in a ripple-carry manner,
  amenable to a simple folding.
  We use an implementation
  similar to the one done in~\cite{Quipper}.
\begin{Verbatim}[fontsize=\relsize{-2},commandchars=\\\{\},
  codes={\catcode`$=3\catcode`^=7\catcode`\#=8}]
adder_aux : (bit $\times$ [bit]) $\to$ (bit $\times$ bit) $\to$ (bit $\times$ [bit])
adder_aux $\langle$w, cs$\rangle$ $\langle$a, b$\rangle$ = let $\langle$w', c'$\rangle$ = bit_adder w a b in $\langle$w', c'::cs$\rangle$

adder : [bit] $\to$ [bit] $\to$ [bit]
adder x y = $\pi#2$ (foldl adder_aux $\langle\ffalse$, nil$\rangle$ (zip y x))
\end{Verbatim}
\end{example}

\subsection{Reversible circuits from operational semantics}
\label{sec:op-sem-rev-circ}

We consider the language \PCF{} as a \define{specification language} for
boolean reversible circuits in the following sense: A term of type
$x_1:\bit,\ldots,x_n:\bit\vdash M:\bit^m$ computes a boolean function
$f_M : \bit^n\to \bit^m$.

In this section, we propose an operational semantics for the language
\PCF{} generating Landauer embeddings, as
described in Section~\ref{sec:reversible-circuits}. The circuit is
produced during the execution of an abstract machine and partial
evaluation of terms. Essentially, a term reduces as usual, except for
the term constructs handling the type $\bit$, for which we only record
the operations to be performed. Formally, the definitions are as
follows.

\begin{definition}\rm
  A \define{circuit-generating abstract machine} is a tuple consisting
  of (1) a typing judgment $p_1:\bit,\ldots,p_{n+k}:\bit\vdash
    M:\bit^m$ ;~~
  (2) a partial circuit $RC :=
    (\{1,\ldots,n\},C) $ where $C$ is a
    list of gates;~~
  (3) a one-to-one linking function mapping the free variables $p_i$
    of $M$ to the wires $\wires{C}\cup\{1,\ldots,n\}$.

  Intuitively, $\{1,\ldots,n\}$ is the set of input wires. The set of
  output wires is not yet computed: we only get it when $M$ is a
  value.  If $G$ is a gate, we write $G::(I,C)$ for the partial
  circuit $(I,G::C)$.
  Given a judgment
  $p_1:\bit,\ldots,p_n:\bit\vdash M:\bit^m$, the empty machine is $
  (M,(\{1,\ldots,n\},\{\}),\{p_i \mapsto i~~\bor~~ i = 1\ldots
  n\}) $ and is denoted with $\emptyam(M)$. The size of the domain of
  a linking function $L$ is written $\sharp(L)$.

  By abuse of notation, we shall write abstract machine with terms,
  and not typing judgements. It is assumed that all terms are
  well-typed according to the definition.
\end{definition}

\begin{definition}\rm
  Given a linking function $L$, a \define{first-order extension of
    $L$} consists of a term of shape
  $
  M~{:}{:}{=}~p_i\bor\prodterm{M_{1},\ldots M_{n}}\bor
  [M_{1},\ldots M_{n}],
  $
  where the $p_i$'s are in the domain of $L$. We say that two
  first-order extensions of $L$ have \define{the same shape} provided
  that they are both products with the same size or lists with the
  same size such as their components have pairwise the same shape.
\end{definition}

\begin{table}[tb]
  \[\begin{array}{r@{~}l}
        (C[(\lambda x.M)N], RC, L) &\rwam (C[M{[}N/x{]}], RC, L)
        \\[.5ex]
        (C[\pi_i\prodterm{M_1,M_2}], RC, L) &\rwam (C[M_i], RC, L)
        \\[.5ex]
        (C[\letunitterm{\punit}{M}], RC, L)  &\rwam (C[M], RC, L)
        \\[.5ex]
        (C[\splitlist\,M], RC, L) &\rwam (C[M], RC, L)
        \\[.5ex]
        (C[\match{\inj_i(P)}{x_1}{M_1}{x_2}{M_2}], RC, L)
        &\rwam
        (C[M_i{[}P/x_i{]}], RC, L)
        \\[.5ex]
        (C[Y(M)], RC, L)
        &\rwam
        (C[M(Y(M))], RC, L)
    \end{array}
    \]
  \caption{Rewrite rules for circuit-generating abstract-machine:
    generic rules.}
  \label{tab:ambetagen}
\end{table}

\begin{table}[t]
  {\begin{align*}
  (C[\ffalse],RC,L)\rwam(C[p_{i_0}], RC, L')
  ~~&~~
  (C[\ttrue],RC,L)\rwam(C[p_{i_0}], (\notgate{i_0})::RC,
  L')
  \\
  (C[{\tt not}~p_i],RC,L)&\rwam(C[p_{i_0}],\cnotgate{i_0}{\ffalse^i}::RC,
  L')
  \\
  (C[{\tt and}~p_i~p_j],RC,L)&\rwam(C[p_{i_0}],
  \cnotgate{i_0}{\ttrue^i\ttrue^j} :: RC, L')
  \\
  (C[{\tt xor}~p_i~p_j],RC,L)&\rwam(C[p_{i_0}],\cnotgate{i_0}{(i,\ttrue)}::
  \cnotgate{i_0}{\ttrue^j}::RC, L')
  \\
  (C[\iftermx{p_i}{V}{W}], RC, L)&\rwam
                                   \\
     &\hspace{-10ex}
  \left\{
    \begin{array}{l@{\quad}l}
      (C[U], RC',L'') & \textrm{$V$ and $W$ of the same shape}
      \\
      C[\errorlist] & \textrm{otherwise}
    \end{array}\right.
  \end{align*}}
\caption{Rewrite rules for circuit-generating abstract-machines: rules
for booleans}
\label{tab:rw-abm}
\end{table}

The set of circuit-generating abstract machines is equipped with a
rewrite-system $(\rwam)$ defined using a notion of
\define{beta-context} $C[-]$, that is, a term with a hole, as follows.
\[
\begin{array}{l}
[-]
\bor
\lambda x.C[-] \bor (C[-])N \bor M(C[-]) \bor
\prodterm{C[-],N} \bor
\prodterm{M,C[-]} \bor
\\
\pi_1(C[-]) \bor
\pi_2(C[-]) \bor
\letunitterm{C[-]}{N} \bor
\letunitterm{M}{C[-]} \bor
\iftermx{C[-]}{N}{P} \bor
\\
\iftermx{\!M\!}{C[\!-\!]}{\!P\!} \bor
\iftermx{\!M\!}{\!N\!}{C[\!-\!]} \bor
\inj_1(C[\!-\!]) \bor 
\inj_2(C[\!-\!]) \bor
\\
\match{C[-]}{x}{M}{y}{N} \bor
\match{P}{x}{C[-]}{y}{N} \bor
\\
\match{P}{x}{M}{y}{C[-]} \bor
  Y(C[-]).
\end{array}
\]
The constructor $[-]$ is the \define{hole} of the context. Given a
context $C[-]$ and a term $M$, we define $C[M]$ as the
variable-capturing substitution of the hole $[-]$ by $M$.

The rewrite rules can then be split in two sets. The first set
concerns all the term constructs unrelated to the type $\bit$. In
these cases, the state of the abstract machine is not modified, only
the term is rewritten. The rules, presented in
Table~\ref{tab:ambetagen}, are the same as for the small-step
semantics of Table~\ref{tab:betared}: apart from the two rules
concerning ${\tt if}$-${\tt then}$-${\tt else}$, all the others are
the same.

The second set of rules concerns the terms dealing with the type
$\bit$, and can be seen as partial-evaluation rules: we only record
in the circuit
the operations that would need to be done. The rules are shown
in Table~\ref{tab:rw-abm}.  The linking function $L'$ is
$L\cup\{p_{i_0}\mapsto i_0\}$, where $i_0$ is a new wire. The
variable $p_{i_0}$ is assumed to be fresh. 
For the case of the ${\tt if}$-${\tt then}$-${\tt else}$, we assume
$V$ and $W$ are first-order extensions of $L$ with the same shape. The
term $U$ is a first-order extension of $L$ with the same shape as $V$
and $W$ containing only (pairwise-distinct) free variables and mapping
to new distinct garbage wires. $L''$ is $L$ updated with this new
data.  Suppose that $V$ contains the variables $v_1,\ldots v_k$, that
$W$ contains the variables $w_1,\ldots w_k$ and that $U$ contains the
variables $u_1,\ldots u_k$.  Then $RC'$ is $RC$ with the following
additional series of gates:
$\cnotgate{u_j}{\ttrue^{p_i}\ttrue^{v_i})}$ and
$\cnotgate{u_j}{\ffalse^{p_i}\ttrue^{w_i}}$.

\begin{remark}
  Note that the set $I$ is never modified by the rules
\end{remark}

Safety properties hold for this new semantics, in the sense that the
only error uncaught by the type system is the term $\errorlist$ that
might be spawned.

\begin{theorem}[Type preservation]
  If $p_1:\bit,\ldots,p_{\sharp(L)}:\bit\vdash M:\bit^m$, if
  $(M,RC,L)$ is an abstract machine and if $(M,RC,L)\rwam
  (N,RC',L')$, then we have the judgement 
  $p_1:\bit,\ldots,p_{\sharp(L')}:\bit\vdash
  N:\bit^m$.\qed
\end{theorem}

\begin{theorem}[Progress]
  Suppose that $p_1:\bit,\ldots,p_{\sharp(L)}:\bit\vdash M:\bit^m$ is valid and that $(M,RC,L)$ is an abstract
  machine. Then either $M$ is a value, or $M$ contains $\errorlist$, or
  $(M,RC,L)$ reduces through ($\rwam$).\qed
\end{theorem}

\subsection{Simulations}

The abstract machine $M$ generates a circuit computing the same
function as the small-step reduction of $M$ in the following sense.

\begin{definition}
  Let $(M,(I,C),L)$ be an abstract machine. 
We write $C(M,(I,C),L)$ for the circuit defined as $(I,
  C, {\rm Range}(L))$.
  Let $(v_k)_{k\in {\rm Range}(L)}$ be the execution of the circuit
  $C(M,(I,C),L)$ on the valuation $\vec{u} = (u_i)_{i\in I}$. We
  define $T(M,(I,C),L)(\vec{u})$ as the term $M$ where each free
  variable $x$ has been replaced with $v_{L(x)}$.
\end{definition}

Intuitively, if $(M,RC,L)$ is seen as a term where some boolean
operations have been delayed in $RC$, then $T(M,RC,L)$ corresponds to the
term resulting from the evaluation of the delayed operations.

\begin{theorem}\label{th:sim}
  \label{th:eq-am-beta}
  Consider a judgment $x_1:\bit,\ldots,x_n:\bit\vdash M:\bit^m$ and
  suppose that
  $
    \emptyam(M)\rwam^*(\prodterm{p_{i_1},\ldots p_{i_k}}, (I,C),
  L).
  $
  Then $k=m$, and provided that $\vec{u} = (b_i)_{i\in I}$, the term
  $T(\prodterm{p_{i_1},\ldots p_{i_k}}, (I,C), L)(\vec{u})$ is
  equal to $\prodterm{c_1,\ldots c_m}$ if and only if the
  term $ \letprodterm{x_1,\ldots x_n}{\prodterm{b_1,\ldots b_n}}{M} $
  reduces to $\prodterm{c_1,\ldots c_m}$.
\end{theorem}

The proof is done using an invariant on a single step of the rewriting
of abstract machines, stated as follows.

\begin{lemma}\label{lem:inv1}
  Consider a judgment $x_1:\bit,\ldots,x_n:\bit\vdash M:\bit^m$ and
  suppose that
  $
    (M,(I,C),L)\rwam(N,(I,C'),L').
  $
  Let $\vec{u}=(u_i)_{i\in I}$ be a valuation.
  Then either the term
  $T(M,(I,C),L)(\vec{u})$ is equal to $T(N,(I,C'),L')(\vec{u})$ if the rewrite corresponds to the
  elimination of a boolean $\ttrue$ or $\ffalse$, or $
  T(M,(I,C),L)(\vec{u}) \to T(N,(I,C'),L')(\vec{u})$, or $N$ contains the
  error term 
  $\errorlist$.
  \qed
\end{lemma}

\begin{proof}[Proof of Theorem~\ref{th:sim}]
  If
  $\emptyam(M)\rwam^*(\prodterm{p_{i_1},\ldots p_{i_k}}, (I,C), L)$,
  then there is a sequence of intermediate rewrite steps where none of
  the terms involved is the term $\errorlist$. From Lemma~\ref{lem:inv1}, one
  concludes that for all valuations $\vec{u}$ on $I$,
  $T(\emptyam(M))(\vec{u}) \to^* T(\prodterm{p_{i_1},\ldots p_{i_k}}, (I,C),
  L)(\vec{u})$.
  Choosing $\vec{u} = (b_i)_{i\in I}$, 
  $T(\emptyam(M))(\vec{u})$ is the term $M$ where each free
  variable $p_{i_j}$ has been substituted with its corresponding
  boolean $b_{i_j}$. 
  Similarly,
  $T(\prodterm{p_{i_1},\ldots p_{i_k}}, (I,C), L)$ is equal to the value$
  \prodterm{b_{i_1},\ldots b_{i_k}}$. We can conclude the proof by
  remarking that the term  $\letprodterm{x_1,\ldots
    x_n}{\prodterm{b_1,\ldots b_n}}{M} $ reduces to  $M$ where each of the free
  variables $p_{i_j}$ have been substituted with $b_{i_j}$, that is, 
  the term $T(\emptyam(M))(\vec{u})$.
\end{proof}

One would have also hoped to have a simulation result in the other
direction, stating that if a (closed) term $M:\bit^m$ reduces to a
tuple of booleans, then $\emptyam(M)$ generates a circuit computing
the same tuple. Unfortunately this is not the case, and the reason is
the particular status of the type $\bit$ and the way the {\tt if}-{\tt
  then}-{\tt else} behaves.

\begin{remark}\label{rem:if1storderbis}
  Let us re-visit the first-order constraint of the {\tt if}-{\tt
    then}-{\tt else} discussed in Remark~\ref{rem:if1storder} in the
  light of this operational semantics. Here, this test behaves as a
  regular boolean operator acting on three arguments: they need to be
  all reduced to values before continuing. This test is ``internal''
  to the circuit: both branches are evaluated during a run of the
  program.  Because it is ``internal'', the type of the branches have
  to be ``representable'': thus the constraint on first-order. This
  test does not control the execution of the program: its
  characteristic only appears at circuit-evaluation time.
  
  With this operational semantics, it is also interesting
  to note that there are two kinds of booleans: the ``internal'' type
  $\bit$, and the ``external'' type defined e.g. as $\bool =
  \tunit\oplus\tunit$. If the former does not control the flow, the
  latter does with the {\tt match} constructor. And unlike {\tt if}-{\tt
    then}-{\tt else}, {\tt match} does not have type
  constraints on its branches.
  
  The term $\errorlist$ can be explained in the light of this
  discussion. Thanks to the condition on the shape of the output
  branches of the test, it is used to enforce the fact that $\bit$
  cannot be coerced to a $\bool$. Indeed, consider the term
  $\iftermx{b}{\nil}{[\ttrue]}$: using a {\tt match} against the
  result of the test, it would allow one to use the bit $b$ for
  controlling the shape of the rest of the circuit. 
  As there is not such construct for reversible circuits, it therefore
  has to be forbidden: it is not possible to control the flow of
  execution of the program through the type $\bit$. And the fact that
  a well-typed term can produce an error is simply saying that the
  type-system is not ``strong enough'' to capture such a problem. It
  is very much related to the fact that the {\tt zip} operator on
  lists cannot be ``safely'' typed without dependent types.
\end{remark}

\section{Internalizing the abstract machine}
\label{sec:internalizing}

Instead of defining an external operational semantics as we did in
Section~\ref{sec:op-sem-rev-circ}, one can internalize the definition
of circuits in the language \PCF{}. Given a program,
provided that one chooses a reduction strategy, one can simulate the
abstract-machine semantics inside \PCF{} using a generic {\em monadic
  lifting}, close to what was proposed in~\cite{MLmonad}.

\subsection{Monadic lifting}
\label{sec:monadic-lifting}

Before going ahead with the full abstract-machine semantics, we present the
monadic lifting of \PCF{} for a monadic function-type. It is the
transposition of Haskell's monads to our language \PCF{}. The main
characteristic of the reversible abstract-machine is to change the
operational behavior of the type $\bit$: the terms $\ttrue$,
$\ffalse$, the inline $\bit$-combinators and the term construct
$\iftermx{\!\textrm{-}\!}{\!\textrm{-}\!}{}$ do not reduce as regular lambda-terms. Instead,
they trigger a side-effect, which can be simulated within a
monad.

\begin{definition}\rm
  A \define{monad} is a function-type $\mathcal{M}(-)$ together with
  two terms $\monadreturn_{\mathcal{M}}^A : A \to \mathcal{M}(A)$ and
  $\monadapp_{\mathcal{M}}^{A,B} : \mathcal{M}(A) \to (A \to
  \mathcal{M}(B)) \to \mathcal{M}(B)$.
  A \define{reversible-circuit monad} is a monad together with a type
  $\wiretype$ and the terms
  $\monadttrue_{\mathcal{M}},
  \monadffalse_{\mathcal{M}}:\mathcal{M}(\wiretype)$,
  $\monadif^{A}_{\mathcal{M}}:\wiretype\to\mathcal{M}(A)\to\mathcal{M}(A)\to\mathcal{M}(A)$,
  and
  $\monadnot^{A}_{\mathcal{M}}:\mathcal{M}(\wiretype\to\mathcal{M}(\wiretype))$,
  and finally $\monadand^{A}_{\mathcal{M}},\monadxor^{A}_{\mathcal{M}} : 
  \mathcal{M}(\wiretype\times\wiretype\to\mathcal{M}(\wiretype))
  $.
\end{definition}

\begin{definition}\rm
  Given a reversible-circuit monad $\mathcal{M}$, we inductively define the
  \define{$\mathcal{M}$-monadic lifting of a type $A$}, written
  $\Lift{M}(A)$. We omit the index
  $\mathcal{M}$ for legibility.
  {\begin{align*}
    \Lift{}(\bit) 
    &= 
    \wiretype,
    &
    \Lift{}(\tunit) 
    &= 
    \tunit,
    \\
    \Lift{}(A\to B)
    &=
    \Lift{}(A) \to \mathcal{M}(\Lift{}(B)),
    &
    \Lift{}(A\times B)
    &=
    \Lift{}(A) \times\Lift{}(B),
    \\
    \Lift{}(A\oplus B)
    &=
    \Lift{}(A) \oplus \Lift{}(B),
    &
    \Lift{}(\listtype{A})
    &=
    \listtype{\Lift{}(A)}.
  \end{align*}}
  The \define{$\mathcal{M}$-monadic lifting of a term $M$}, denoted with
  $\Lift{M}(M)$, is defined as follows. First, we set $\Lift{}(\ttrue) = \monadttrue$,
    $\Lift{}(\ffalse) 
    =
    \monadffalse$, 
    $\Lift{}({\tt and}) 
    =
    \monadand$,
    $\Lift{}({\tt xor}) 
    =
    \monadxor$ and 
    $
    \Lift{}({\tt not}) 
    =
    \monadnot$. Then
  {\begin{align*}
    \Lift{}(x) 
    = 
    \monadreturn_{\mathcal{}}~x,
    \hspace{-9ex}~~&~~\hspace{9ex}
    \Lift{}(\punit) 
    = 
    \monadreturn_{\mathcal{}}~\punit,
    \\
    \Lift{}(\lambda x.M)
    = 
    \monadreturn_{\mathcal{}}~ \lambda x.\Lift{}(M),            
    \hspace{-9ex}~~&~~\hspace{9ex}
    \Lift{}(\tt split)
    =
    \monadreturn_{\mathcal{}}~ \lambda x.\monadreturn_{\mathcal{}}~({\tt split}\,x),
    \\
   \Lift{}(MN)
    &= 
    \monadapp_{\mathcal{}}~\Lift{}(M)~
    \lambda x.
    \monadapp_{\mathcal{}}~\Lift{}(N)~
    \lambda y.xy,
    \\
    \Lift{}(\prodterm{M,N})
    &= 
    \monadapp_{\mathcal{}}~\Lift{}(M)~
    \lambda x.
    \monadapp_{\mathcal{}}~\Lift{}(N)~
    \lambda y.\monadreturn_{\mathcal{}}~\prodterm{x,y},
    \\
    \Lift{}(\pi_i(M))
    &=
    \monadapp_{\mathcal{}}~\Lift{}(M)~
    \lambda x.\monadreturn_{\mathcal{}}~\pi_i(x),
    \\
    \Lift{}(\inj_i(M))
    &=
    \monadapp_{\mathcal{}}~\Lift{}(M)~
    \lambda x.\monadreturn_{\mathcal{}}~\inj_i(x),
    \\
    \Lift{}(\letunitterm{M}{N})
    &= 
    \monadapp_{\mathcal{}}~\Lift{}(M)~
    \lambda x.
    \monadapp_{\mathcal{}}~\Lift{}(N)~
    \lambda y.\monadreturn_{\mathcal{}}~\letunitterm{x}{y},
    \\
   \Lift{}(\match{P}{z_1\!\!}{\!\!M}{z_2\!\!}{\!\!N})=\hspace{-15ex}~
    &\\
     &\hspace{-10ex}
    \monadapp_{\mathcal{}}~\Lift{}(P)~
    \lambda x.
    \match{x}{z_1\!\!}{\!\!\Lift{}(M)}{z_2\!\!}{\!\!\Lift{}(N)},
    \\
    \Lift{}(Y(M))
    &=\\&\hspace{-10ex}\monadapp_{\mathcal{}}~\Lift{}(M)~
    \lambda f.
    \monadreturn~(Y(\lambda y.\lambda z.\monadapp\,(y\,\punit)\,f))\punit
    \\
    \Lift{}(\iftermx{\!P\!}{\!M\!}{\!N\!})
    & =
    \monadapp_{\mathcal{}}~\Lift{}(P)~
    \lambda x.
    ((\monadif_{\mathcal{}}~x)~\Lift{}(M))~\Lift{}(N)
  \end{align*}}
\end{definition}

\begin{remark}\label{rem:cbv}
  Note that in this definition of the lifting, we followed a
  call-by-value approach: the argument $N:\Lift{M}(A)$ of a function
  $M:A \to \Lift{M}(B)$ is first reduced to a value before being fed
  to the function. This will be discussed in Section~\ref{sec:cbv}.
\end{remark}

The fact that a monad is equipped with $\monadttrue$, $\monadffalse$,
$\monadxor$, $\monadand$, $\monadnot$ and $\monadif$ is not a
guarantee that the lifting will behave as expected. One has to choose
the right monad for it. It is the topic of
Section~\ref{sec:rev-circ-from-monad}.  However, in general this
monadic-lifting operation preserves types (proof by induction on the
typing derivation).

\begin{theorem}
  Provided that $x_1:A_1,\ldots,x_n:A_n\vdash M:B$ is valid, so is the judgment
  $
  x_1:\Lift{M}(A_1),\ldots,x_n:\Lift{M}(A_n)\vdash
  \Lift{M}(M):\mathcal{M}(\Lift{M}(B)).
  $\qed
\end{theorem}

\subsection{Reversible circuits from monadic lifting}
\label{sec:rev-circ-from-monad}

All the structure of the abstract machine can be encoded in
the language \PCF{}.
A wire is a natural number. A simple way to represent them
is with the type $\wiretype ~{:}{=}~ \listtype{\tunit}$. The number
$0$ is the empty list while the successor of $n$ is
$(\cons{\punit}{n})$.
A gate is then $\gatetype ~{:}{=}~
\wiretype\times\listtype{\wiretype\times\bit}$.
A raw circuit is  $\listtype{\gatetype}$.

We now come to the abstract machine. In the formalization of
Section~\ref{sec:op-sem-rev-circ}, we were using a state with a circuit and a
linking function. In this internal representation, the linking
function is not needed anymore: the computation directly acts on
wires. However, the piece of information that is still needed is the
next fresh value. The state is encapsulated in $\statetype ~{:}{=}~
\listtype{\gatetype}\times\wiretype$.
Finally, given a type $A$, we write $\circtype(A)$ for the type
$\statetype \to (\statetype\times A)$: this is a computation
generating a reversible circuit.

The type operator $\circtype(-)$ can be equipped with the structure of
a reversible-circuit monad, as follows. First, it is obviously a
state-monad, making the two first constructs automatic:
\[
\monadreturn :=
  \lambda x.\lambda s.(s,x)
\textrm{ and }
  \monadapp
  :=
  \lambda xf.\lambda s.\letprodterm{s',a}{x\,s}{(f\,a)\,s'}.
\]
The others are largely relying on the fact that \PCF{} is expressive
enough to emulate what was done in Section~\ref{sec:op-sem-rev-circ}.
Provided that $\succterm$ stands for the successor function, we can
$\monadffalse$ as the lambda-term
$\lambda
      s.\letprodterm{c,w}{\!\!s}{\prodterm{\prodterm{c,\succterm\,w},w}}$
and $\monadttrue$ as  the lambda-term
$\lambda
      s.\letprodterm{c,w}{\!\!s}{\prodterm{\prodterm{\prodterm{w, \nil} {:}{:} c,\succterm\,w},w}}
$.
Note how the definition reflects the reduction rules corresponding to
$\ttrue$ and $\ffalse$ in Table~\ref{tab:rw-abm}: in the case of
$\ffalse$, the returned wire is the next fresh one, and the state is
updated by increasing the ``next-fresh'' value by one unit. In the
case of $\ttrue$, on top of this we add a not-gate to the list of
gates in order to flip the value of the returned wire.
The definitions of $\monadnot$, $\monadand$ and $\monadxor$ are
similar. For $\monadif$, one capitalizes on the fact that we know the
structure of the branches of the test, as they are of first-order
types. One can then define a zip-operator $A^0\times A^0\to A^0$ for
each first-order type $A^0$.

\subsection{Call-by-value reduction strategy}
\label{sec:cbv}

As was mentioned in Remark~\ref{rem:cbv}, the monadic lifting
intuitively follows a call-by-value approach. It can be formalized by
developing a call-by-value reduction strategy for circuit-abstract
machines.
Such a definition follows the one of the reduction proposed in
Section~\ref{sec:op-sem-rev-circ}: we first design a notion of {\em
  call-by-value evaluation context} $E[-]$ characterizing the
call-by-value redex that can be reduced.

\begin{table}
  \[\def\nl{\\[1ex]}
    \begin{array}{r@{}l}
    (E[(\lambda x.M)V], RC, L) &\cbvam (E[M{[}V/x{]}], RC, L)
    \nl
    (E[\pi_i\prodterm{V_1,V_2}], RC, L) &\cbvam (E[V_i], RC, L)
    \nl
    (E[\letunitterm{\punit}{M}], RC, L)
    &\cbvam (E[M], RC, L)
      \nl
    (E[\splitlist\,V], RC, L) &\cbvam (E[V], RC, L)
                                \nl
   (E[\match{\inj_i(V)}{x_1}{M_1}{x_2}{M_2}], RC, L)
    \hspace{-10ex}\nl &\cbvam
    (E[M_i{[}V/x_i{]}], RC, L)
    \nl
    (E[Y(\lambda x.M)], RC, L)
    &\cbvam
    (E[M[Y(\lambda x.M)/x]], RC, L)
    \end{array}
  \]
  \vspace{-2ex}
  \caption{Call-by-value for circuit-generating abstract-machine:
    generic rules.}
  \label{tab:amcbv}
\end{table}

\begin{table}
  {\begin{align*}
  (E[\ffalse],RC,L)\cbvam(E[p_{i_0}], RC, L')
  ~~&~~
  (E[\ttrue],RC,L)\cbvam(E[p_{i_0}], (\notgate{i_0})::RC,
  L')
  \\
  (E[{\tt not}~p_i],RC,L)&\cbvam(E[p_{i_0}],\cnotgate{i_0}{\ffalse^i}::RC,
  L')
  \\
  (E[{\tt and}~p_i~p_j],RC,L)&\cbvam(E[p_{i_0}],
  \cnotgate{i_0}{\ttrue^i\ttrue^j} :: RC, L')
  \\
  (E[{\tt xor}~p_i~p_j],RC,L)&\cbvam(E[p_{i_0}],\cnotgate{i_0}{(i,\ttrue)}::
  \cnotgate{i_0}{\ttrue^j}::RC, L')
  \\
  (E[\iftermx{p_i}{V}{W}], RC, L)&\cbvam
                                   \\
&\hspace{-10ex}
  \left\{
    \begin{array}{l@{\quad}l}
      (E[U], RC',L'') & \textrm{$V$ and $W$ of the same shape}
      \\
      \errorlist & \textrm{otherwise}
    \end{array}\right.
  \end{align*}}
\caption{Call-by-value for circuit-generating abstract-machines: rules
for booleans}
\label{tab:cbv-abm}
\end{table}

\begin{definition}
A \define{call-by-value context} is a beta-context with the following
grammar
\[
\begin{array}{lll}
E[-]
&{:}{:}{=}&
[-]
\bor
(E[-])N \bor V(E[-]) \bor
\prodterm{E[-],N} \bor
\prodterm{V,E[-]} \bor
\\
&&
\pi_1(E[-]) \bor
\pi_2(E[-]) \bor
\letunitterm{E[-]}{N} \bor
\letunitterm{V}{E[-]} \bor
\iftermx{E[-]}{N}{P} \bor
\\
&&
\iftermx{V}{E[-]}{P} \bor
\iftermx{V}{W}{E[-]} \bor
\inj_1(E[-]) \bor 
\\
&&
\inj_2(E[-]) \bor
\match{E[-]}{x}{M}{y}{N} \bor
Y(E[-]).
\end{array}
\]
We define the \define{call-by-value reduction strategy} on
circuit-generating abstract machines as shown in Tables~\ref{tab:amcbv}
and~\ref{tab:cbv-abm}.
\end{definition}

\begin{remark}
  Essentially, the generic rules of Table~\ref{tab:ambetagen} are
  turned into their call-by-value version in the standard way.  For
  example, we require that
  $(E[(\lambda x.M)V], RC, L) \cbvam (E[M{[}V/x{]}], RC, L)$ happens
  only when $V$ is a value. Similarly, the rules of
  Table~\ref{tab:rw-abm} are reflected in Table~\ref{tab:cbv-abm},
  replacing $C[-]$ with $E[-]$.

  Remark however that the reduction strategy does not exactly match
  the rewrite system $(\rwam)$ in one special case:the rewrite rule
  concerning the fixpoint. We chose to modify it in order for the
  fixpoint to behave in a call-by-value manner: we embedded two steps
  of $(\rwam)$ into one step of the strategy.
  One can then show that if $(M,RC,L)\cbvam(M',RC',L')$
  then $(M,RC,L)\rwam^*(M',RC',L')$.
\end{remark}

In the light of this reduction strategy and of the monadic lifting of
the previous section, one can now formalize what was mentioned in
Remark~\ref{rem:cbv}. 
First, one can turn an abstract machine into a lifted term.

\begin{definition}\label{def:amlift}
  Let $x_1:\bit,\ldots,x_{n+k}:\bit\vdash M:B$ and let $(M,(C,I),L)$
  be an abstract machine where $I=\{1\ldots n\}$. Then we define
  $\Lift{}(M,(C,I),L)$ as the term\vspace{-2ex}
  \[\vspace{-2ex}
    \left(\Lift{}(M)[\overline{L(x_{n+1})}/x_{n+1}\ldots
    \overline{L(x_{n+k})}/x_{n+k}]\right)
    \prodterm{\overline{C},\succterm\,\overline{\max({\rm
        Range}(L))}},
  \]
  where $\overline{C}$ is the representation of $C$ as a term of type
  $\listtype{\gatetype}\times\wiretype$, and where $\overline{n}$ with
  $n$ an integer is the representation of $n$ as a term of type
  $\listtype{\unittype}$.
\end{definition}

Then, provided that $\simeq_\beta$ stands for the reflexive, symmetric
and transitive closure of the beta-reduction on terms and choosing $M$
and $(M,(C,I),L)$ as in Definition~\ref{def:amlift}:

\begin{theorem}
  Suppose that $(M,RC,L)\cbvam(M',RC',L')$. Then 
  $\Lift{}(M,RC,L)$ is beta-equivalent to $\Lift{}(M',RC',L')$.
  \qed
\end{theorem}

Provided that the beta-reduction is confluent, this essentially says
that the abstract-machine semantics can be simulated with the monadic
lifting.

\begin{corollary}\label{cor}
  If ${}\vdash M:\bit^m$ and 
  $
    \emptyam(M)\cbvam(\prodterm{x_1,\ldots x_m},(C,I),L),
  $
  then the term $\pi_1(\Lift{}(M))$ is beta-equivalent to $\overline{C}$,
  where $\overline{C}$ is the representation of $C$ as a term, as
  described in Definition~\ref{def:amlift}.
  \qed
\end{corollary}

\section{Synthesis of quantum oracles}
\label{sec:use-cases}
\label{sec:reduc-size-circ}

A rapid explanation is needed here: In quantum computation, one does
not deal with classical bits but with the so-called \define{quantum
  bits}. At the logical level, a quantum algorithm consists of one or
several \define{quantum circuits}, that is, reversible circuits with
quantum bits flowing in the wires.

Quantum algorithms are used to solve classical problems. For example:
factoring an integer~\cite{shor}, finding an element in a unordered
list~\cite{grover}, finding the solution of a system of linear
equations~\cite{QLSA}, finding a triangle in a graph~\cite{TF}, {\em etc}.
In all of these algorithms, the description of the problem is a
completely non-reversible function $f:\bit^n \to \bit^m$ and it has to
be encoded as a reversible circuit computing the function
$\bar{f}:\bit^n\times\bit^m \to \bit^n\times\bit^m$ sending $(x,y)$ to
$(x, y\,{\tt xor}\,f(x))$, possibly with some auxiliary wires set back
to $0$.

\begin{wrapfigure}{r}{0.4\textwidth}
\vspace{-30pt}
\begin{minipage}{0.39\textwidth}
\scriptsize$
\begin{array}{@{}c@{}}
\xy
0;<1em,0em>:
(0,0.5);(2,0.5)**\dir{-};
(2,2)**\dir{-};
(2,3)**\dir{.};
(2,4.5)**\dir{-};
(0,4.5)**\dir{-};
(0,3)**\dir{-};
(0,2)**\dir{.};
(0,0.5)**\dir{-};
(1,2.5)*{T_f};
(0,1);(-1,1)*!R{{0}}**\dir{-};
(0,4);(-1,4)*++!R{x}**\dir{-};
(2,1);(3.5,1)*{\bullet}="a"**\dir{-};(5,1)**\dir{-};
(2,4);(5,4)**\dir{-};
(7,4);(8,4)*++!L{x}**\dir{-};
(-1,2.5)*+!R{{0\cdots0}};
(5,0.5);(7,0.5)**\dir{-};
(7,2)**\dir{-};
(7,2)**\dir{.};
(7,4.5)**\dir{-};
(5,4.5)**\dir{-};
(5,3)**\dir{-};
(5,2)**\dir{.};
(5,0.5)**\dir{-};
(6,2.5)*{T_f^{\textrm{-}1}};
(7,1);(8,1)**\dir{-};
(-1,-0.5)*++!R{y};(3.5,-0.5)*{\oplus}="b"**\dir{-};(8,-0.5)*++!L{y\,
  {\tt xor}\, f(x)}**\dir{-}
\multiline{0}{2}{-1}{}
\multiline{2}{2}{5}{}
\multiline{7}{2}{8}{}
\save
(8.3,3).(8.3,1)!C*+\frm{\}}
\restore
\save
"a";"b"**\dir{-};
(8,2)*+++!L{\begin{array}{@{}l@{}}\textrm{back to}\\\textrm{state $0$}\end{array}}
\restore
\endxy
\end{array}
$
\end{minipage}
\vspace{-20pt}
\end{wrapfigure}
A canonical way to produce such a circuit is with a {\em Bennett
  embedding}. The procedure is shown on the right. First the
Landauer embedding $T_f$ of $f$ is applied. Then the output of the
circuit is {\tt xor}'d onto the $y$ input wires, and finally the
inverse of $T_f$ is applied. In particular, all the auxiliary wires
are back to the value $0$ at the end of the computation.

The method we propose in this paper offers a procedure for generating
the main ingredient of this construction: the Landauer embedding.
One just has to encode the problem in the language $\PCF$ (or
extension thereof), possibly test and verify the program, and generate
a corresponding reversible circuit through the monadic lifting.
Theorems~\ref{th:eq-am-beta} and~\ref{cor} guarantee that
the monadic lifting of the program will give a circuit computing
the same function as the original program.

This algorithm was implemented within the language Quipper, and used
for non-trivial oracles~\cite{Quipper,PLDI}. Note that Quipper is not
the only possible back-end for this generic monadic lifting: nothing forbids us
to write a back-end in, say, Lava~\cite{Claessen-2001}.

\begin{example}\label{ex:addercircuit}
  The first example of code we saw (Example~\ref{ex:adder}) computes
  an adder. One can run this code to generate a reversible adder:
  Figure~\ref{fig:adder-4-raw} shows the circuit generated when fed
  with 4-bits integers. One can see 4 blocks of pairs of
  similar shapes. 
\begin{figure}[tb]
  \centering
  \includegraphics[width=4in]{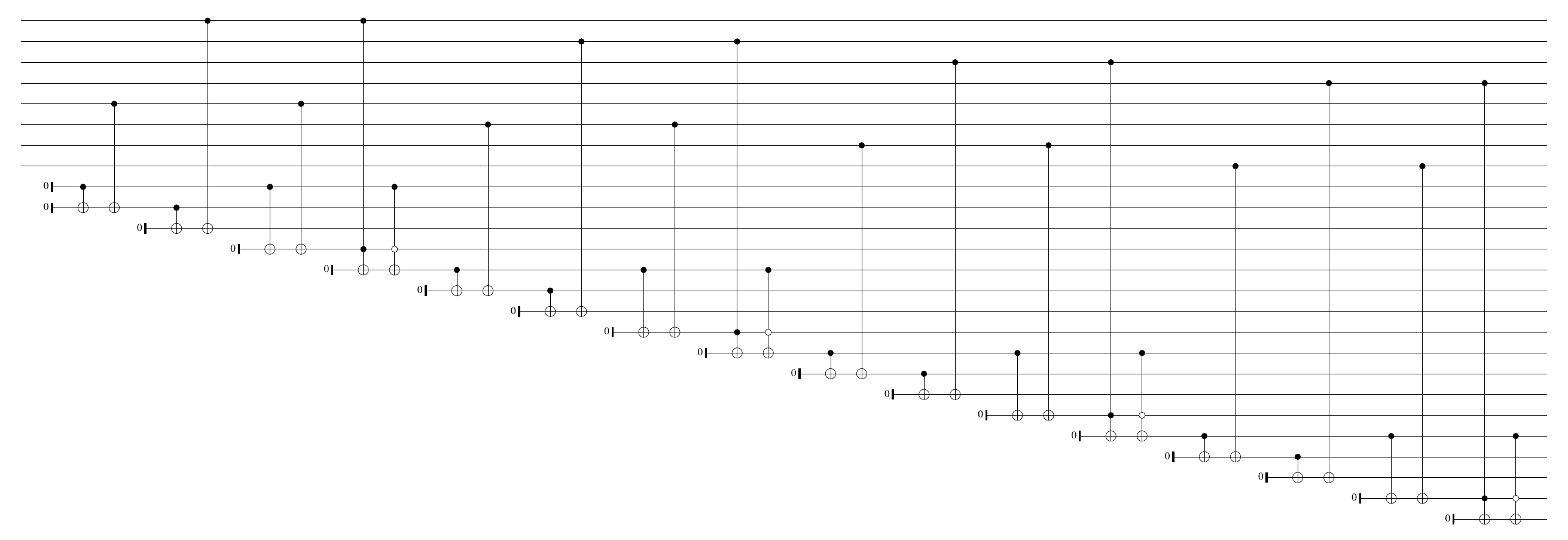}
  \caption{Reversible adder for 4-bit integers.}
  \label{fig:adder-4-raw}
\end{figure}
\end{example}
\begin{example}\label{ex:sin}
  In the oracle for the QLSA algorithm~\cite{QLSA,QLSAnous}, one has to solve a
  system of differential equations coming from some physics problem
  using finite elements method. The bottom line is that it involves
  analytic functions such as sine and atan2.

  Using fixed-point real numbers on 64 bits, we wrote a sine function
  using a Taylor expansion approximation. In total, we get a
  reversible circuit of 7,344,140 multi-controlled gates (with
  positive and negative controls). The function atan2 was defined
  using the CORDIC method. The generated circuit contains 34,599,531
  multi-controlled gates. These two functions can be found in
  Quipper's distribution~\cite{Quipper}.
\end{example}

\subsection{Efficiency of the monadic lifting}
The monadic lifting proposed in this paper generates circuits that are
efficient in the sense that the size of a generated circuit is linear
in the number of steps it takes to evaluate the corresponding
program. This means that any program running in polynomial time upon
the size of its input generates a polynomial-sized circuit. Without
any modification or optimization whatsoever, the technique is
therefore able to generate an ``efficient'' circuit for an arbitrary,
conventional algorithm. This is how the circuit for the function sine
cited in Example~\ref{ex:sin} was generated: first, a conventional
implementation was written and tested. When ready the lifting was
performed, generating a circuit.

\subsection{Towards a complete compiler}
Compared to other reversible compilers~\cite{revs}, the approach taken
in this paper considers the construction of the circuit as a process
that can be completely automatized: the stance is that it should be
possible to take a classical, functional program with conventional
inductive datatypes and let the compiler turn it into a reversible
circuit without having to interfer (or only marginally). We do not
claim to have a final answer: we only aim at proposing a research path
towards such a goal.

A first step towards a more complete compiler for \PCF{} would involve
optimization passes on the generated circuits.  Indeed, as can be
inferred from a quick analysis of Figure~\ref{fig:adder-4-raw}, if
monadic lifting generates efficient circuits it does not produces
particularly lean circuits.

\subsection{Towards optimization of generated circuits}
\label{sec:optim}

There is a rich literature on optimization of reversible
circuits~\cite{rev-survey-2011,revkit,maslov-templates-iccad03}, some
also considering positive and negative
controls~\cite{soeken13white}. All of these works are relevant for
reducing the size of the circuits we get.

The purpose of this section is not to aggressively optimize the
circuits we get from the monadic lifting, 
but instead to reason on the particular shapes we obtain from
this monadic semantics. If the reduction of a lambda-term into a
reversible circuit is so verbose, it is partly due to the fact that
garbage wires are created for every single intermediate result. We aim
at pointing out the few optimization rules stemming from the circuit
generation and reflecting these low-level optimizations to
high-level program transformations.

The reversible adder of Figure~\ref{fig:adder-4-raw} is very
verbose. By applying the simple optimization schemes presented
in this section, one gets the smaller circuit of
Figure~\ref{fig:adder-4-optim}. One clearly sees the carry-ripple
structure, and it is in fact very close to known reversible
ripple-carry adders (see e.g.~\cite{feynman85quantum}). In the following
discussion, we hint at how program transformations could be applied in
order to get a circuit of compactness similar to the one obtained from
the low-level circuit rewrites.

\smallskip
\noindent
{{\bf Algebraic optimizations.}}~
Let us consider the example of the 1-bit adder of
Figure~\ref{fig:optim}, from the code of Example~\ref{ex:adder}.
\begin{figure}[tb]
  \centering
  \includegraphics[width=3in]{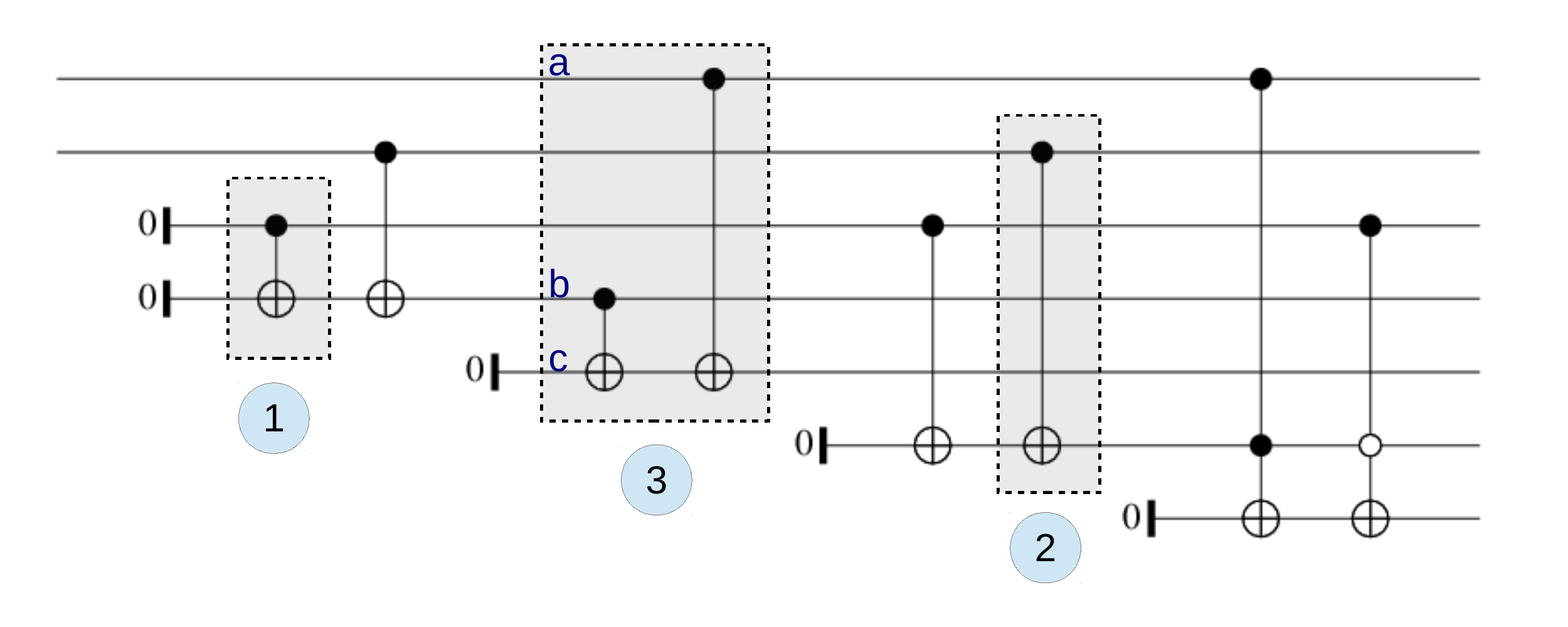}
  \caption{Adder of 1-bit integers, with potential optimizations highlighted}
\label{fig:optim}
\end{figure}
Three simple potential optimizations are highlighted.

In general, these optimizations require to have a knowledge of the
value of the bits flowing in the wires (e.g. Dashed Box 2). Since
there are input wires, this information needs to be kept in algebraic
form, as a function of the input wires. Of course, for non-trivial
circuits this means actually computing the circuit.

However, because of the shape of the generated circuit, instead of a
complete algebraic form, for the purpose of circuit simplification it
is often enough to keep only partial algebraic information about the
wires. To each piece of wire, we essentially keep a limited knowledge
of the past operations.

\smallskip
\noindent
{\it Dashed Box 1.}~ Gates acting on wires of known constant value. The
gate will never fire as the control will always be negative. The
gate can be removed.

From the perspective of the code the circuit comes from, this
situation typically occurs when constant booleans are manipulated, for
example with the term ${\tt xor}\,\ffalse\,\ffalse$.

\smallskip
\noindent
{\it Dashed Box 2.}~ Copy of one wire to another one.
Provided that the controlled wire is never controlled later on, one
can remove the gate and move all the controls and gates of the bottom
wire to the top wire.

The fact that the wire is not used later on means that the particular
intermediate result is never used again: From the point of view of the
program it means that this particular result is only used {\em
  linearly}. The typical case where this occurs is in a term such as
${\tt and}\,({\tt not}\,({\tt and}\,x\,y))\,z$. A garbage wire 
is created to hold the result of the ${\tt not}$, but this is not
needed as this intermediate result is not going to be reused. Instead
one can directly apply a not-gate on the result of ${\tt and}\,x\,y$.

\smallskip
\noindent
{\it Dashed Box 3.}~ Here, $c = a\,{\tt xor}\,b$. The two CNOTs can
be replaced with only one connecting wires $a$ and $b$, and one
could have removed wire $c$ altogether.
Again, some care must be taken: the new active wire should not be
controlled later on (ruling out wire $a$), and the controls and
actions of wire $c$ should be moved to the new active wire.

This situation can also be understood as a linearity constraint on the
program side.

\smallskip
\noindent
{\bf Reduction strategies and garbage wires.}~
The call-by-value reduction strategy we follow sometimes computes
unused intermediate results, therefore generating gates acting on
unused wires. One can safely discard such gates.

Note that the abstract machine is agnostic to the choice of reduction
strategy. In general, depending on the chosen reduction path, the generated
circuits do not have the same size and shape. Consider
the term
$
x:\bit\vdash (\lambda y.{\tt and}~y~y)({\tt not}~x):\bit. 
$
A call-by-value 
reduction strategy first evaluate the argument,
and then feeds the output
\begin{wrapfigure}{r}{0.3\textwidth}
\vspace{-20pt}
{\scriptsize$
\xymatrix@=.1cm{
  x\ar@{-}[r]&*{\circ}\ar@{-}[r]\ar@{-}[d] &x
  \\
  0\ar@{-}[r]&*{\oplus}\ar@{-}[r] & {\tt not}~x\ar@{-}[r]
  &*{\bullet}\ar@{-}[r]\ar@{-}[d] & {\tt not}~x
  \\
  &&
  0 \ar@{-}[r]&*{\oplus}\ar@{-}[r] & {\tt not}~x.
}
$}
\vspace{-20pt}
\end{wrapfigure}
value to the {\tt and} operator. Since having
the same variable $y$ means that the two controlling wires
collapsed, we get the circuit presented on the right:
${\tt not}~x$ is the output of the circuit, which is the meaning of 
the lambda-term.

With a call-by-name
strategy, the term instead
reduces to the following circuit. The
\begin{wrapfigure}{r}{0.45\textwidth}
\vspace{-20pt}
{\scriptsize$
\xymatrix@=.1cm{
  x\ar@{-}[r]&*{\circ}\ar@{-}[r]\ar@{-}[d] &*{}\ar@{-}[r]
  &*{\circ}\ar@{-}[r]\ar@{-}[dd] & x
  \\
  0\ar@{-}[r]&*{\oplus}\ar@{-}[r] & *{}\ar@{-}[r]
  &*{}\ar@{-}[r] & {\tt not}~x \ar@{-}[r]
  &*{\bullet}\ar@{-}[r]\ar@{-}[dd] & {\tt not}~x 
  \\
  &&
  0 \ar@{-}[r]&*{\oplus}\ar@{-}[r] & {\tt not}~x \ar@{-}[r]
  &*{\bullet}\ar@{-}[r] & {\tt not}~x
  \\
  &&&&0\ar@{-}[r]
  &*{\oplus}\ar@{-}[r] & ({\tt not}~x)~{\tt and}~({\tt not}~x)\hspace{-1.5cm}
}\hspace{1.5cm}
$}
\vspace{-20pt}
\end{wrapfigure}
last
wire is the output wire, the other wires are garbage
wires. Of course, we get the same result: $({\tt not}~x)~{\tt
  and}~({\tt not}~x)$ and ${\tt not}~x$ are indeed the same boolean
value. However, note that the circuit is different than the one
generated by a call-by-value strategy. In general, call-by-name tends
to generate larger circuits as arguments are duplicated and evaluated
several times. The case where it is not true is when the argument of a
function is not used: in call-by-value, the argument would generate a
piece of circuit, whereas in call-by-name, since it would not be
evaluated, the argument would leave no trace on the circuit.

\begin{figure}[tb]
  \centering
  \begin{subfigure}[b]{0.4\textwidth}
    \includegraphics[width=\textwidth]{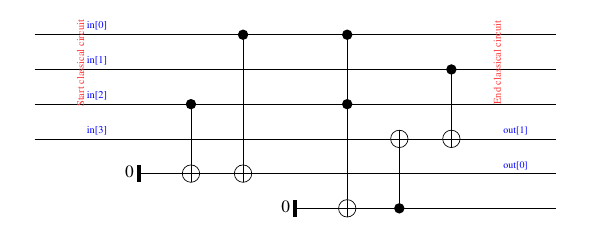}
    \caption{Circuit waiting for shuffling.}
    \label{fig:shuffling-1}
  \end{subfigure}%
  ~
  \begin{subfigure}[b]{0.4\textwidth}
    \includegraphics[width=\textwidth]{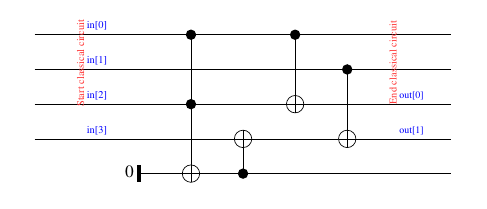}
    \caption{After the shuffling.}
    \label{fig:shuffling-2}
  \end{subfigure}%
  \caption{Circuit optimized with and without shuffle.}
  \label{fig:shuffling}
\end{figure}

\smallskip
\noindent
{\bf Optimizations by shuffling.}~
A less obvious circuit modification is to send CNOT gates as far as
possible to the right, by swapping the order of gates. This is again a
side-effect of the particular shape of the generated circuit.

If it does not decrease the size the circuit, it is able to reveal
places where algebraic optimizations can be performed. For example,
consider the circuit in Figure~\ref{fig:shuffling-1}. The two first
CNOTs can be moved to the far right-end, becoming a hidden
CNOT (as Dashed Box 2 of Figure~\ref{fig:optim})\,:
they are merged into one single CNOT and the first auxiliary wire
is removed. We get the circuit in Figure~\ref{fig:shuffling-2}.

The corresponding program transformation modifies the term
\[
\lettermx{z_1}{{\tt xor}\,x\,y}{\lettermx{z_2}{f(x,y)}{g(z_1,z_2)}}
\]
to
\[\lettermx{z_2}{f(x,y)}{}\lettermx{z_1}{{\tt xor}\,x\,y}{
g(z_1,z_2)}.
\]
As $g$ does not use $x$ nor $y$, one of the algebraic optimization
might apply.

\begin{figure}[tb]
  \centering
  \includegraphics[width=2in]{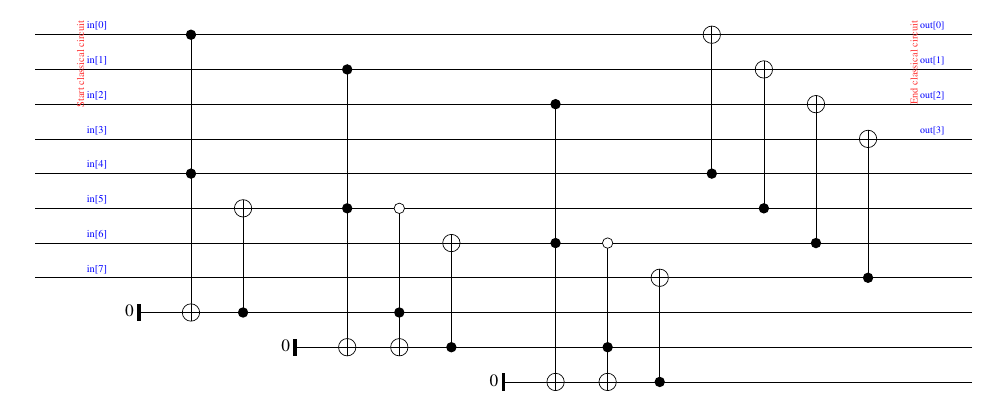}
  \caption{Reversible adder for 4-bit integers, optimized.}
  \label{fig:adder-4-optim}
\end{figure}

\begin{example}
  By applying these optimization schemes on the reversible adder of
  Figure~\ref{fig:adder-4-raw}, one gets the circuit presented in
  Figure~\ref{fig:adder-4-optim}. One can now clearly see the
  carry-ripple structure, and it is in fact very close to known
  reversible ripple-carry adders (see
  e.g.~\cite{feynman85quantum}). These optimizations were implemented
  in Quipper: applied on larger circuits such as the ones of
  Example~\ref{ex:sin}, we get in general a size reduction by a factor
  of 10.
\end{example}

\section{Conclusion and future work}
\label{sec:conclusion}
In this paper, we presented a simple and scalable mechanism to turn a
higher-order program acting on booleans into into a family of
reversible circuits using a monadic semantics. The main feature of
this encoding is that an automatically-generated circuit is guaranteed
to perform the same computation as the original program.
The classical description we used is a small PCF-like language, but it
is clear from the presentation that another choice of language can be
made. In particular, an interesting question is whether it is possible
to use a language with a stronger type system for proving
properties of the encoded functions.

A second avenue of research is the question of the parallelization of
the generated circuits. The circuits we produce are so far completely
linear. Following the approach in~\cite{ghica07gos}, using parallel
higher-order language might allow one to get parallel reversible
circuits, therefore generating circuits with smaller depths.

Finally, the last avenue for research is the design of generic
compiler with a dedicated code optimizations.  Indeed, an analysis of
the specific optimizations described in Section~\ref{sec:optim}
suggests that these could be designed at the level of code, therefore
automatically generating leaner circuits up front. This opens the door
to the design of specific type systems and code manipulations in a
future full compiler.  and back-end, specific circuit optimizations.

\end{document}